\documentclass{article}
\usepackage{graphicx} % Required for inserting images
\usepackage{amsmath,amsthm,amssymb}
\usepackage{authblk}
\usepackage{hyperref}
\usepackage{comment}
\newtheorem{observation}{Observation}
\newtheorem{proposition}{Proposition}

\newtheorem{lemma}{Lemma}
\newtheorem{remark}{Remark}
\newtheorem{corollary}{Corollary}
\newtheorem{definition}{Definition}

\usepackage{fullpage}

\title{Searcher Competition in Block Building\thanks{Christoph Schlegel and Danning Sui were working at Flashbots at the time of writing this paper. There is no any financial or other commercial interest into publishing this paper.}}
\author[1]{Akaki Mamageishvili}
\author[2]{Christoph Schlegel}
\author[3]{Benny Sudakov}
\author[2]{Danning Sui}

\affil[1]{Offchain Labs}
\affil[2]{Flashbots}
\affil[3]{ETH Z\"{u}rich}

\date{September, 2025}

\begin{document}

\maketitle

\begin{abstract}
We study the amount of maximal extractable value (MEV) captured by validators, as a function of searcher competition, in blockchains with competitive block building markets such as Ethereum. We argue that the core is a suitable solution concept in this context that makes robust predictions that are independent of implementation details or specific mechanisms chosen.
We characterize how much value validators extract in the core and quantify the surplus share of validators as a function of searcher competition. Searchers can obtain at most the marginal value increase of the winning block relative to the best block that can be built without their bundles. Dually this gives a lower bound on the value extracted by the validator. If arbitrages are easy to find and many searchers find similar bundles, the validator gets paid all value almost surely, while searchers can capture most value if there is little searcher competition per arbitrage.
%For the case of passive block-proposers we study, 
Moreover, mechanisms that implement core allocations in dominant strategies, for submodular values, there is a unique dominant-strategy incentive compatible core-selecting mechanism that gives each searcher exactly their marginal value contribution to the winning block. We extend our model to multiple concurrent proposers in which, under mild assumptions, the core is empty.  

We validate our theoretical prediction empirically with aggregate bundle data and find a significant positive relation between the number of submitted backruns for the same opportunity and the median value captured by the proposer from the opportunity.    
\end{abstract}

\section{Introduction}

Blockchains that support smart contracts frequently run decentralized finance (DeFi) applications. This in turn gives rise to the phenomenon of miner/maximal extractable value,~\cite{mev_original}:
Blockchain protocols give {\it validators}, sometimes called proposers, the right to order transactions for a particular block. This effectively means that validators have a local monopoly on including or excluding transactions or ordering transactions in a particular way, in order to generate value for themselves from this privileged position. 
However, since extracting value from transaction ordering in an optimal way is a difficult task, more specialized actors, such as {\it block builders} and (arbitrage) {\it searchers} participate in the value extraction process in smart contract blockchains such as Ethereum. Block builders aggregate different arbitrage opportunities, liquidations or "sandwiches" in one block, using transactions from the public mempool, different private mempools, order flow auctions, or other transaction sources together with their own transactions. Then, they bid against other builders, to get their block published in a canonical chain of blocks. The bid is paid to the current block proposer. 
Arbitrage opportunities are typically found by more specialized players, {\it searchers} and passed to the builders.  See Figure~\ref{BB}.
\begin{figure}\includegraphics[scale=0.22]{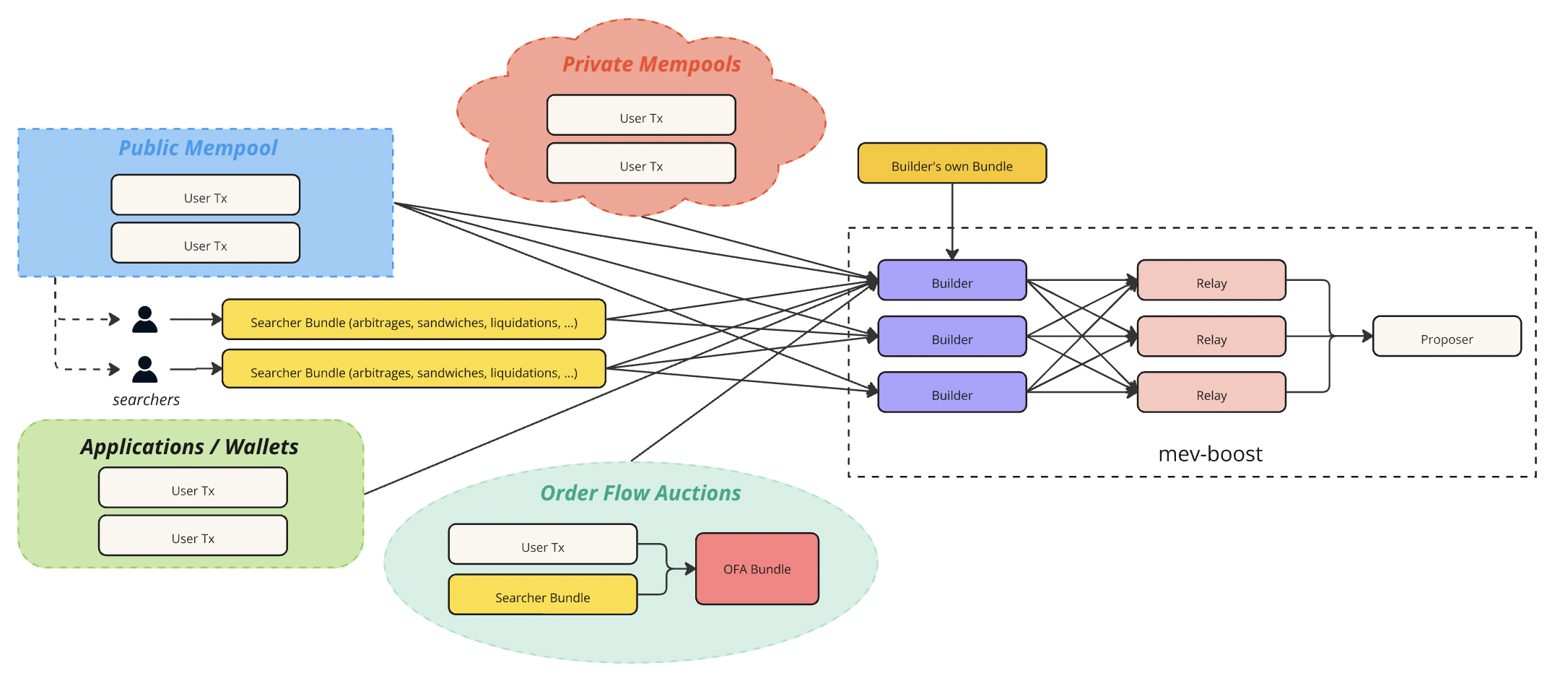}
\label{BB}
\caption{A schematic representation of the Ethereum block building process.}
\end{figure} for a schematic diagram of the Ethereum block building process. The resulting strategic interaction between searchers, builders and the validator lead to a value distribution of the arbitrage gains between these players. Empirical evidence shows that the majority of the total observable MEV is captured by validators.~\footnote{See \url{https://www.galaxy.com/insights/perspectives/distribution-of-mev-surplus/}.}

In this paper, we analyze the competition between searchers in order to explain the value distribution between the proposer (or block builder) and the searchers. Alternatively, our model can also be applied to the competition of different order flow providers and explain the value distribution between different order flow providers and the proposer (or block builder). In our model, different searchers may capture different value from different opportunities and may or may not find unique opportunities that are not found by competing searchers. This competition and specialization should, in turn, explain the value capture by different searchers in the MEV supply chain. Our model abstracts away from the intermediate layer of block building where block builders aggregate searcher bundles to blocks and compete in a bidding procedure. 
However, the ability for different combinations of searchers or order flow providers to jointly generate different blocks is captured implicitly by the solution concept we focus on: We consider core allocation where different combinations of players can construct blocks of different values and distribute the value among themselves. The core constraints require that no coalition of players should be able to generate more value for themselves (produce a more valuable block) than they get paid in the realized allocation. Thus, while we generically talk about searchers, our model can, for example, also capture the case where some searchers act as block builders.
The core is a natural solution concept in the context of block building and MEV capture as it captures the competition and possibility of collaboration of different players while abstracting away from any particular mechanisms that would intermediate between searchers, builders and proposers. Thus, it provides robust predictions on the set of plausible value distributions that can arise in any sufficiently competitive block building market.\footnote{The ability to make predictions about value distributions could be particularly useful in situations where profits are not directly observed. It has, for example, been observed that block builders with access to exclusive order flow tend to capture more value in the block building market. However, it cannot necessarily be inferred from on chain data how much of the profit the builder actually passes on to the order flow provider, since side payments might not be directly observable. Our model would generally predict that an order flow provider with unique flow should capture most of the value they contribute.} 
As a next step, one can look into mechanisms that implement particular core allocations by eliciting privately held information of searchers about the value they can generate from different blocks. 
All values in the model are assumed to be calculated after paying chain (gas) fees. One extra way to have heterogeneous values across searchers is their ability to reduce (gas) fees they pay. 

 Our first result shows that searchers can at most obtain their marginal contribution to the winning block (the difference in value between the best block that can be built with their transactions and without their transactions). If the value is submodular (value in the set extension has diminishing returns), then giving each searcher their marginal contribution is in the core.\footnote{While transactions submitted to a builder might exhibit significant complementarities, submodularity can be argued to be a reasonable assumption in a sufficiently consolidated market. Order flow providers with complementary flow have an incentive to integrate to capture more value together. Thus, the realizable value might be submodular in the contributions of the different (consolidated) players.}
A straightforward calculation shows that in a world with passive block producers, this particular core allocation coincides with the Vickrey-Clark-Growes (VCG) outcome (\cite{AGT}) and hence can be implemented in dominant strategies (Corollary~\ref{implementation}). On the other hand (Proposition~\ref{implementation2}), no other core-selecting mechanism is dominant-strategy incentive-compatible. This result gives additional justification for further study of the searcher-optimal point in the core: the value obtained by the validator in this point can be interpreted as the maximal extractable value he can obtain taking into account information rents captured by the searchers.

As a next step, we consider a stochastic model of searcher competition to study the value distribution between the searchers and the validator. Different opportunities (arbitrage or MEV) have a fixed probability of being found by a searcher. The more searchers find an opportunity, the more value can be extracted by the validator and the less value can be extracted by the searchers. We show that if the probability of finding an opportunity is bounded from below by $p>\log(n)/n$, where $n$ is the number of searchers actively searching in the strategy, with high probability the validator captures all value in all core allocations. On the other hand, if the probability of sucessfully searching is very low, $p\in\Theta(1/n)$, then with positive constant probability, searchers will be able to capture all value in the searcher-optimal core allocation.

Next, we empirically validate our stochastic model. We obtain bundle data from the MEV-Share Order Flow Auction (OFA) that allows to study the relation between searcher competition and validator profit. We find that there is significant positive relation between the number of bundles that backrun the same user transaction and the median profit of the validator from opportunities with the same number of backruns. 

We extend our model to multiple concurrent proposers (MCP) setting, a recent proposal for block building, which improves economic censorship resistance of the chain, check out~\cite{MCP} and discussion within. In this setting, there are multiple different entities that bring their sets of transactions or bundles, from which the block is formed. We again abstract away from the exact rules how the block is formed from the inputs of proposers. As long as there are at least two searchers, we show that if all but one proposer can include all searcher-relevant transactions in the block, the core of the underlying game is empty. The assumption is reasonable if there are sufficiently many proposers and searcher bundles do not require too much space in the block. This result hinges on the economic instability of multiple concurrent proposers.
However, if there is only one searcher, there is a trivial core where all the value goes to the only searcher as the proposers substitute each one of them perfectly.

\subsection{Glossary of Terms}\label{sec:glossary}

In this section, we will list terms that are used throughout the introduction and later in the paper:

\begin{enumerate}
    \item {\bf Miner, Validator, Proposer:} A designated party that builds the next block of the blockchain. Initially, the term {\it miner} was used, referring to the fact that the next block builder was supposed to find a nonce resulting into a suitable hash of its own block, an activity called "mining" in the Proof-of-Work context of Bitcoin,~\cite{bitcoin}. Later, in the Proof-of-Stake context, the role of a block builder was changed to the term {\it validator}, also internalizing the functions of validating and voting for the block proposed by the block-proposing entity. In the later evolving block building market of the Ethereum chain, validators usually act as a {\it proposer} that signs a block without actually building it, while the parties that actually build blocks are called {\it builders}. 
    \item{\bf Arbitrage:} Opportunity in which risk-free profits can be made. Usually, arbitrages arise when the same asset is priced differently on multiple markets. In the context of blockchains, the prices between centralized exchange and on-chain exchanges differ, creating arbitrage opportunities in which the arbitrageur sells (buys) over (under) priced asset on chain and buys (sells) under (over) priced asset on the centralized exchange. Different arbitrageurs, however, have different portfolios both on-chain and centralized exchanges, which results into different values they can extract.%However, in the context of blockchains, there are more types of arbitrages.
    \item {\bf Searcher:} Party specialized in finding arbitrage opportunities. There are searchers that find many different arbitrage types, also searchers that specialize in specific arbitrage types, introducing heterogeneity among them.
    \item {\bf Frontrunning:} Blockchains usually maintain a pool of pending transactions, before they are included and executed in the block(s). If the pool of pending transactions is public, that allows searchers (or proposers/builders) to take advantage by inserting their own transaction in front of a publicly available transaction in the block for execution, leading to a profit.
    \item {\bf Backrunning:} Similar to frontrunning, backrunning refers to inserting the own transaction after a public transaction in the block for execution, leading to a profit for a backrunning party. Backrunning is usually considered a harmless activity, usually fixing the price movement by noise traders, unlike frontrunning.
    \item {\bf Sandwich attack:} Refers to putting the own transactions in front and after a publicly available transaction in the block for execution. This is done by bundling all 3 transactions into one transaction (bundle). Typically, the victim transaction performs an exchange on-chain. The transaction inserted in front increases the price of the asset that the victim transaction buys, and then the following transaction reverts the trade direction. This results in a higher price slippage for a victim transaction sender and gains for an attacker.
    \item {\bf Liquidation:} In many decentralized finance applications, collateral is provided, for example, to borrow some assets. When a collateral asset price goes down, it is usually sold at a disadvantaged price to a depositor, creating a value for a potential buyer. Similarly to arbitrage opportunities, liquidation brings different value to different searchers depending their portfolios.
\end{enumerate}

\subsection{Related Literature}
The phenomenon of miner extractable value has first been documented in~\cite{mev_original}. Another early contribution on front-running in decentralized finance is~\cite{eskandarisok}. MEV has been documented by a variety of public dashboards and data sets~\footnote{For the examples see \url{https://mevboost.pics/},\url{https://libmev.com/}}. We refer to~\cite{flash-bot} for the discussion on searchers, builders and validators in MEV extraction, for the Proof-of-Work setting.

Recently, also the magnitude of non-atomic arbitrage has been empirically investigated, where searchers realize one lag of a trade on chain and one on a different domain, see e.g.~\cite{heimbach2024non}~\footnote{Also check out  \url{https://dune.com/flashbots/lol}.}.

The block building market structure that has evolved in Ethereum since the change to proof of stake, has been the topic of several recent contributions: The dashboard~\url{https://orderflow.art/} documents empirically the supply chain through which transaction requests land on chain. 
\cite{MEV_centralization} studies whether MEV and Proof-of-Stake rewards capture leads to centralization, discussing both validator and builder roles in the market formation. 
\cite{pai2023structural} argue that in the current market structure searchers and builders have an incentive to vertically integrate.~\cite{capponi2024proposer} argue that the builder market is prone to centralization.~\cite{decentralization} provide descriptive statistics on the level of decentralization on the builder landscape.~\cite{who_gets_MEV} proposes a dynamic MEV sharing mechanism that the authors argue results into better decentralization and fair allocation. In our paper, we focus on neither the fairness of the value distribution nor long-term centralization threats due to asymmetric value distribution.

\cite{bahrani2024} studies questions of implementation with active block producers. Our model is similar to theirs, but we focus on the question of implied value distribution. Our positive results for dominant-strategy incentive compatible core-selecting mechanisms for passive block builders complement their negative result for active block builders. \cite{bahrani2024} also study the role of searchers as intermediaries.

The notion of core was first formally defined in~\cite{core_original}. For an overview of results around the core and submodularity, see, e.g.~\cite{Moulin}. Similar considerations as in our paper appear in the literature on core-selecting auctions~\cite{day2007fair,day2008core}.

\section{Examples and Model}

In this section, we define the game and cooperative game theory terms formally. Before that, we will discuss three illustrative examples to give a motivation and intuition.

\subsection{Examples}

Suppose that there are three searchers, labeled $s_1$, $s_2$ and $s_3$ and three opportunities, labeled $o_1$, $o_2$ and $o_3$. Consider the following three matrices of values searchers derive for each opportunity. It is assumed that opportunities derive their values independently from each other, and the block size is large enough to accommodate extraction of all opportunities.

\begin{center}
\begin{tabular}{ |c|c|c|c| } 
 \hline
S/O & $o_1$ & $o_2$ & $o_3$ \\
\hline

$s_1$ & $1$ & $0$ & $0$ \\
\hline
$s_2$ & $0$ & $1$ & $0$ \\
\hline
$s_3$ & $0$ & $0$ & $1$ \\
 \hline
\end{tabular}
\end{center}
In the example one, all searchers find unique opportunities.

\begin{center}
\begin{tabular}{ |c|c|c|c| } 
 \hline
S/O & $o_1$ & $o_2$ & $o_3$ \\
\hline

$s_1$ & $0$ & $1$ & $1$ \\
\hline
$s_2$ & $1$ & $0$ & $1$ \\
\hline
$s_3$ & $1$ & $1$ & $0$ \\
 \hline
\end{tabular}
\end{center}

In the example two, each opportunity that each searcher finds is found by some other searcher. 

\begin{center}
\begin{tabular}{ |c|c|c|c| } 
 \hline
S/O & $o_1$ & $o_2$ & $o_3$ \\
\hline

$s_1$ & $1$ & $1$ & $0$ \\
\hline
$s_2$ & $0$ & $1$ & $0$ \\
\hline
$s_3$ & $0$ & $0$ & $1$ \\
 \hline
\end{tabular}
\end{center}

In the example three, searcher $s_1$ finds one unique opportunity and one opportunity found by the searcher $s_2$, which only finds one opportunity, while the searcher $s_3$ finds one unique opportunity. 

All these examples differ in many aspects, namely, who finds how many opportunities and how unique these opportunities are. One thing that all examples share is that the total value of all opportunities is $3$. It is therefore interesting how this value can be distributed between searchers and the validator that has an exclusive right to propose the next block, taking into account their bargaining power searchers derive from uniquely finding opportunities. In the following section, we define tools to solve this simple problem with independent valuations and a much more general problem with correlated valuations. 

\subsection{Model}

There is a finite set $\mathcal{S}$ of searchers that submit transactions for inclusion in the block.  Similarly as in~\cite{bahrani2024} we will usually identify searchers with (bundles of) transactions they have sent for inclusion. However, our model also allows for the interpretation that the same searcher (address) sends multiple bundles for inclusion.  
There is one validator (proposer), denoted by $V$. For each set of searchers $A\subseteq \mathcal{S}$ there is a finite set of feasible blocks $\mathcal{B}(A)$ that can be built from bundles of transactions submitted by searchers in $A$. A searcher $i$ generates value $v_i(B)$ from a block $B$ and the validator generates value $v_V(B)$ from block $B$. Our model can capture externalities (searcher $i$'s realized value may not only depend on her included transactions but also other transactions in the realized block) and active validators/block producers (we may have $v_V(B)\neq 0$). We assume that utility is transferable and the final utility realized by searcher $i$ if block $B$ is realized and she makes a payment of $p_i$ (e.g. to the validator) is $v_i(B)-p_i$. 

Since utility is transferable, we can define a coalitional value function $v:2^{\mathcal{S}\cup\{V\}}\rightarrow \mathbb{R}_+$ by 

$$v(S\cup V):=\max_{B\in\mathcal{B}(S)}\left(\sum_{i\in S}v_i(B)+v_V(B)\right),$$
and 
$$v(S)=0 \text{ if }V\notin S,$$
i.e. in case the validator is part of the coalition, the coalitional value is the value of the total welfare maximizing block consisting of transactions from searchers in the coalition, and in case the validator is not part of the coalition, no value can be generated, as the validator is necessary to realize a block. This implicitly assumes a communication mechanism between a validator and searchers, however, the validator does not need to be actively involved into value extraction.

It will be useful subsequently to introduce the short hand notation $$\bar{v}(S):=v(S\cup V)$$
for $S\subseteq\mathcal{S}$ to denote the value that searchers $S$ can generate together with the validator. Observe that by construction, the (collective) value function is monotonic, $$\bar{v}(A)\leq \bar{v}(B)\quad\text{ for all }A\subseteq B\subseteq \mathcal{S};$$
if we receive more bundles to build a block that will increase welfare weakly, since we can always discard submitted bundles when building a block. 

In all three examples, $\bar{v}(\{s_1,s_2,s_3\}) = 3$. In example one, value function is defined as: $\bar{v}(\{s_i\})=1$, for any $i$ and $\bar{v}(\{s_i,s_j\})=2$, for any $i\neq j$. 

In example two, value function is defined as: $\bar{v}(\{s_i\})=2$, for any $i$ and $\bar{v}(\{s_i,s_j\})=3$, for any $i\neq j$.

In example three, value function is defined as: $\bar{v}(\{s_2\})=\bar{v}(s_3)=1$, $\bar{v}(s_1)=2$, $\bar{v}(\{s_1,s_2\})= \bar{v}(\{s_2,s_3\})=2$, and $\bar{v}(\{s_1,s3\})=3$.

We make the following assumption on the (collective) value function:\newline

\noindent {\bf Submodularity}: Let $A,B\subseteq \mathcal{S}$. Then
$$\bar{v}(A)+\bar{v}(B)\geq \bar{v}(A\cup B)+\bar{v}(A\cap B),$$

Submodularity states that the value of a block we can build from transactions submitted by searchers in $A$ and $B$, is bounded by subtracting the value of a block we can build from transactions in both $A$ and $B$ from the sum of values we can achieve from building a block with transactions in $A$ and a block with transactions in $B$.

Submodularity requires that there are not-too-strong complementarities between different submitted bundles. We can justify this assumption in two ways: first, it may be that complementarities are not strong and different MEV opportunities provide value that is mostly independent from other opportunities. Second, it may be that the complementarities are already absorbed by searchers, e.g. in the sense that searchers who provide complementary flow have an incentive to integrate their operations and send their flow together to extract more value. It is worth noting that our upper bound on searcher values holds also for non-submodular value functions, but may be loose in that case. Searchers conducting most common arbitrage opportunities described in section~\ref{sec:glossary} exhibit no strict complementarities. It can be easily checked that all the examples satisfy the submodularity condition.

It is easy to show that for monotone value functions, submodularity is equivalent to requiring decreasing marginal value:
\\\newline
\noindent {\bf Decreasing Marginal Value}: Let $A\subseteq B\subseteq \mathcal{S}$ and $a\in A$. Then, the value function satisfies the following inequality:
$$\bar{v}(B)-\bar{v}(B\setminus\{a\})\leq \bar{v}(A)-\bar{v}(A\setminus \{a\}).$$

A direct consequence of submodularity is the following lemma, which will be useful subsequently:
\begin{lemma}
Let ${A}\subseteq {B}\subseteq\mathcal{S}$. Then, for decreasing marginal value functions, the following inequality holds: 
$$\bar{v}({B})-\bar{v}({A})\geq \sum_{i\in B\setminus A}(\bar{v}({B})-\bar{v}({B}\setminus\{i\})).$$
\end{lemma}
\begin{proof}
We prove the result by induction on $N:=|{B}\setminus {A}|$. For $N=0$ the result holds trivially. Now suppose the result holds for $N\geq 0$ and consider the case  $N+1$. Let $j\in B\setminus A$. By induction assumption
$$\bar{v}({B}\setminus\{j\})-\bar{v}({A})\geq \sum_{i\in B\setminus (A\cup\{j\})}(\bar{v}({B}\setminus\{j\})-\bar{v}({B}\setminus\{i,j\})).$$
Adding $\bar{v}(B)-\bar{v}(B\setminus\{j\})$ on both sides and using submodularity (which implies decreasing marginal value), we obtain
\begin{align*}\bar{v}({B})-\bar{v}(A)&\geq \sum_{i\in B\setminus (A\cup\{j\})}(\bar{v}(B\setminus\{j\})-\bar{v}(B\setminus\{i,j\}))+\bar{v}(B)-\bar{v}(B\setminus\{j\})\\&\geq \sum_{i\in B\setminus A}(\bar{v}(B)-\bar{v}(B\setminus\{i\})).\end{align*}

\end{proof}

We define the allocation of value between all players of the cooperative game.

\begin{definition}
An {allocation} is a value distribution $x: \mathcal{S}\cup V\rightarrow \mathbb{R}_{+}$ such that
$$\sum_{j\in \mathcal{S}\cup V}x_j\leq v(\mathcal{S}\cup V).$$    
\end{definition}
%A payment rule is {\bf budget balanced} when $\sum_{j\in V\cup \mathcal{S}}p(j)\leq T$. 
%From now on, we consider only budget balanced payments.
Note the abuse of notation, we will use $x_j$ instead of $x(j)$ to denote an allocation to a player $j$.
Next, we define the core, which is the main solution concept that we will focus on in this paper:

\begin{definition}
An allocation is in the {core} if the following inequality 

\begin{equation}\label{core_inequality}
    \sum_{j\in C}x_j\geq v(C)
\end{equation} holds for any subset $C\subseteq  \mathcal{S}\cup V$
and all value is distributed:
$$\sum_{j\in V\cup \mathcal{S}}x_j=v(\mathcal{S}\cup V).$$    
\end{definition}

Inequalities~\eqref{core_inequality} indicate that there is no coalition of players $C$ that get allocated less than the value they can derive themselves, $v(C)$. Intuitively, if such a coalition of players existed, they could profitably deviate from the allocation, leading to instability of the allocation.

\begin{remark}
As usual in the formulation of the core, the solution proposes a value allocation without specifying an explicit implementation through a block and payments between the searchers and the validator. Let $B^*\in\mathcal{B}(\mathcal{S})$ be a welfare maximizing block, i.e.~$\sum_{i\in\mathcal{S}\cup V}v_i(B^*)=v(\mathcal{S}\cup V)$. Then a core value allocation $x$ can be implemented by realizing the block $B^*$ and requiring that each individual searcher $i$ makes a payment of $p_i=v_i(B^*)-x_i$ to the validator.
\end{remark}

Immediately from the requirement that without the validator no value can be realized, it follows that 
the validator getting all gains is in the core. 
\begin{observation}\label{obs1}
   The  core is always non-empty. The allocation where $x_V=v(\mathcal{S}\cup V)$ and $x_i=0$ for each $i\in\mathcal{S}$ is in the core.
\end{observation}

\section{Analysis}
Our first main result states each searcher can at most capture their marginal contribution to the realized block, and any allocation that gives each searcher at most their marginal contribution is in the core. In particular, this implies that there is a searcher optimal core allocation (giving each searcher exactly their marginal contribution to the realized block) and a validator optimal allocation (giving the validator all realized value).
\begin{proposition}\label{thm:main}
    An allocation is in the core if and only $$0\leq x_i\leq \bar{v}(\mathcal{S})-\bar{v}(\mathcal{S}\setminus\{i\})$$ for each $i\in \mathcal{S}$ and $x_V=\bar{v}(\mathcal{S})-\sum_{i\in \mathcal{S}}x_i.$
\end{proposition}

\begin{proof}
First, we show that $x_i\leq \bar{v}(\mathcal{S})-\bar{v}(\mathcal{S}\setminus\{i\})$ for each $i$. Suppose not. In that case $$x_V+\sum_{j\in \mathcal{S}\setminus\{i\}}x_j=\bar{v}(\mathcal{S})-x_i<\bar{v}(\mathcal{S}\setminus\{i\}).$$
 The previous inequality 
contradicts the core-stability of $x$ which requires inequality~(1) to hold for $C=\mathcal{S}\setminus\{i\}\cup V$. The lower bound on $x_i$ is trivial. It remains to show that all allocations with $0\leq x_i\leq \bar{v}(\mathcal{S})-\bar{v}(\mathcal{S}\setminus\{i\})$ are in the core. Let $x$ be a vector satisfying these inequalities. Let $A\subseteq\mathcal{S}$ and observe that  $$x_V+\sum_{j\in A}x_j= \bar{v}(\mathcal{S})-\sum_{i\in \mathcal{S}\setminus A}x_i\geq \bar{v}(A),$$
where the last inequality follows from Lemma 1. Thus, the core inequalities~\eqref{core_inequality} are satisfied for all coalitions $A\cup V$ with $A\subseteq\mathcal{S}$. For coalitions without the validator the core inequalities are trivially satisfied, as all searchers get non-negative value in $x$.

\end{proof}

In the example one, the prefix of the core vector $(x_{s_1},x_{s_2},x_{s_3})$ is the product of $[0,1]\times [0,1]\times[0,1]$, while the validator is paid $x_V=3-x_{s_1}-x_{s_2}-x_{s_3}$. 

In the example two, all searchers are getting, i.e., $x_i=0$ for any $i\in \{s_1,s_2,s_3\}$. The validator is paid all the value equal to $3$.

In the example three, the prefix of the core $(x_{s_1}, x_{s_2},x_{s_3})$ is the product of $[0,1]\times 0\times [0,1]$, while the validator is paid $3-x_{s_1}-x_{s_3}$.

\subsection{Implementation}
In reality, the values that searchers obtain from different blocks is private information to them. However, in a world with passive block producers\footnote{We know from~\cite{bahrani2024} that with active block producers implementation of non-trivial solutions is not possible. This in turn resembles previous results from other context where negative results prevail if incomplete information in a two-sided markets is on both sides of the market~\cite[]{Satterthwaite}. In particular, for assignment games~\cite{shapley_shubik}, which can be reinterpreted as the special case of our model where the validator has additively separable value, there is no mechanism that implements a core-allocation in dominant strategies, for any domain of valuations when there is at least one profile of valuations for which a core allocation that gives positive value to some searcher exists. For a proof of this "folk theorem" see, e.g.~\cite{Sotomayor}.}, we can implement the extreme point in the core that gives maximal value to searchers in dominant strategies:
Consider the case where the block producer is passive, that is, $v_V(B)=0$ for each block $B$.
Observe that VCG-payments in this problem are defined by
$$p_i:=\max_{B\in \mathcal{B}(\mathcal{S}\setminus\{i\})}\sum_{j\neq i}v_j(B)-\sum_{j\neq i}v_j(B^*)=\bar{v}(\mathcal{S}\setminus\{i\})-\sum_{j\neq i}v_j(B^*),$$
for each $i\in\mathcal{S}$
where $B^*$ is the welfare-optimal block that can be produced.
A straightforward calculation shows that,
$$v_i(B^*)-p_i=\sum_{j\in\mathcal{S}}v_j(B^*)-\bar{v}(B\setminus\{i\})=\bar{v}(\mathcal{S})-\bar{v}(\mathcal{S}\setminus\{i\}),$$
i.e. the searcher-optimal core outcome coincides with the VCG outcome. We obtain the following corollary of Proposition~\ref{thm:main}:
\begin{corollary}\label{implementation}
Under submodular value and with passive block-producers, the searcher-optimal core-outcome can be implemented in dominant strategies.
\end{corollary}
On the other hand, it is straightforward to see that other selections from the core that not always select the VCG outcome are not dominant-strategy incentive compatible.
\begin{proposition}\label{implementation2}
Under submodular value and with passive block-producers, any core-selecting mechanism that is not always choosing the searcher-optimal outcome in the core is not dominant-strategy incentive compatible.
\end{proposition}
\begin{proof}
Suppose that for reported value functions $(v_i)_{i\in\mathcal{S}}$, a block $B^*$ and payments $(p_i)_{i\in\mathcal{S}}$ are chosen by the mechanism. By core-stability and Proposition~\ref{thm:main}, we have $0\leq v_i({B}^*)-p_i\leq v(\mathcal{S})-v(\mathcal{S}\setminus \{i\})$ for searcher $i$, or equivalently $$v_{i}(B^*)\geq p_i\geq \bar{v}(\mathcal{S}\setminus\{i\})-\sum_{j\neq i}v_j(B^*)$$ Suppose for the sake of contradiction that the last inequality is strict and $i$ reports different values $\tilde{v}_i$ with $\tilde{v}_i({B})=v_i({B})$ for $B$ blocks not including transactions by $i$, and $\tilde{v}_i(B)\leq \tilde{v}_i(B^*)$ for blocks including transactions by $i$ (so that $B^*$ is still optimal) and 
$$p_i> \tilde{v}_i(B^*)>\bar{v}(\mathcal{S}\setminus i)-\sum_{j\neq i}v_j(B^*).$$

Note that this change in value preserves the submodularity of the coalitional value function that the same block $B^*$ is optimal and that the payment $\tilde{p}_i$  for searcher $i$ now satisfies
$$ p_i>\tilde{v}_i(B^*)\geq \tilde{p}_i\geq\bar{v}(\mathcal{S}\setminus i)-\sum_{j\neq i}v_j(B^*),$$
which is stricly less. Therefore, $i$ gains from misreporting.
\end{proof}

The corollary and previous proposition motivate to further study the searcher-optimal point in the core. The value that a (passive) builder/validator obtains in this point is the maximal extractable value taken into account information rents that searchers can capture. In the next sections, we study the particular case where the value of a block is derived from independent (MEV) opportunities for which different searchers compete. For that case, we derive results on when competition lets the core collapses to one point (in which the validator captures all value) and when lack of competition allows searchers to generate positive value in the searcher-optimal core allocation.
\section{Independent Bundles and Competing Searchers}
In this section, we consider the special case of additively separable value where the value of a block is the sum of values derived from the individual (bundles) of transactions. Moreover, we look at a scenario where multiple searchers may compete for the same (arbitrage) opportunities so that their submitted bundles possibly "clash" with bundles submitted by other searchers.\footnote{This matches the reality of searching where often different searchers compete in the same strategy and find conflicting bundles among which the block builder chooses the most profitable one and includes it in the block while discarding the less profitable competing bundles. Around one third of submitted bundles to Ethereum block builders "clash".\label{footnote}}
We denote opportunities by $\mathcal{A}$ and now can identify a block by a matrix $B=(B_{ij})_{i\in \mathcal{A},j\in\mathcal{S}}$ where $B_{ij}=1$ if searcher $j$'s bundle competing for opportunity $i$ is included, $B_{ij}=0$ if it is not included. We require that $\sum_{j\in\mathcal{S}}B_{ij}\leq 1$ so that multiple clashing bundles cannot be included in the same block.
We can then write searcher $j$'s value from block $B$ as
$$v_j(B)=\sum_{i\in\mathcal{A}}v_{ij}B_{ij},$$
where $v_{ij}$ is the value extracted by searcher $j$ from opportunity $i$. 
 We can add additional constraints such as a capacity constraint on the total number of bundles.
In the unconstrained case, where all blocks are feasible, we have
\begin{equation}\label{index_val}
    v(S\cup V)=\max_{B\in\mathcal{B}(S)}\sum_{i\in\mathcal{A}}\sum_{j\in\mathcal{S}}v_{ij}B_{ij}=\sum_{i\in\mathcal{A}} \max_{j\in S}v_{ij}.
\end{equation}
%We denote the valuation of arbitrage $i$ by searcher $j$ with $D_{i,j}\geq 0$. The matrix of valuations is denoted by $D$. 
%We assume that for each slot $i$ there are several searchers $j=1,\ldots,n_i$ with bundles $t_{i,j}$ competing on the slot $i$.

To get an intuition how the core looks like in this case, we start this section with few examples and observations. Previously, in Observation~\ref{obs1} we had observed that giving all value to the validator is always in the core. In the opposite direction, if there is no searcher competition, i.e. no two searchers find the same opportunity, the validator can receive no value in the core.

\begin{observation}
    There are examples of core allocations where the validator receives $0$.
\end{observation}  

\begin{proof}
     Suppose that for each opportunity $i\in\mathcal{A}$ there is at most one searcher submitting a bundle of positive value, $v_{ij}>0$ for at most one $j$. Then, the value allocation with $x_{j}=\sum_{i:v_{ij}>0}v_{ij}$ for each searcher $j$ and $x_{V}=0$ is in the core.
\end{proof}

Next, we give an example where the maximum payment to the validator is enforced in the core.

\begin{observation}
    There are examples where the validator gets the full value of the winning block in any core allocation.
\end{observation} 

\begin{proof}
Suppose for each opportunity $i\in \mathcal{A}$ at least two searchers submit the same highest value bundle (or no searcher finds the opportunity). Then for each searcher $j$ we have $x_{j}=v(\mathcal{S})-v(\mathcal{S}\setminus\{j\})=0$. The claim follows from Proposition~\ref{thm:main}.
\end{proof} 

%For each coalition of players, the value is calculated in the following way. $v(S)=0$ for any $S\in \mathcal{S}$. That is, any coalition of searchers, without validator, achieves value $0$. On the other hand, for any $S\in \mathcal{S}$, we have $v$

%\begin{equation}\label{value_function}
 %  v(V\cup S)=\sum_{i\in A}\max_{j\in S}D_{i,j}.
%\end{equation} 

%That is, a coalition of searchers together with a validator derives a maximum value for each arbitrage among the coalition members. Let $T$ denote maximum possible gain by searchers and the validator, $T:=v(V\cup \mathcal{S})$.

For a vector of non-negative numbers $X$, let $SH(X)$ denote the second highest coordinate of it.  Then, let $M$ denote the following sum:

\begin{equation}\label{lower_bound}
  M:=\sum_{i\in \mathcal{A}}M_i:=\sum_{i\in\mathcal{A}}SH((v_{ij})_{j\in\mathcal{S}}).
\end{equation}

Additive value function as defined in~\eqref{index_val} are submodular. Thus, we obtain the following special case of Proposition~\ref{thm:main}.

%Let $M_{j,i}$ denote the value $$M_{j,i}:=\max(0, D_{i,j}-SH(D_i)),$$
%and $M_j:=\sum_{i=1}^{m} M_{j,i}$. The first value denotes how much searcher $S_j$ contributes towards arbitrage $A_i$, while the second value denotes searcher's total contribution.
%Next, we prove the main result of this section, characterizing the polytope of the core payment rules.

%Theorem~\ref{thm:main} in particular implies: 

\begin{corollary}\label{value}
The allocation in which the validator receives $x_V=M$ and searcher $j$ receives $$x_j=\sum_{i:j\in\text{argmax}_{j\in \mathcal{S}} v_{ij}}(v_{ij}-M_i)$$ is in the core. 
\end{corollary}

This particular core allocation is the worst for the validator and the best for searchers and can be implemented, see Corollary~\ref{implementation}, in dominant strategies by a generalized second price auction for bundles.

\subsection{Stochastic Model}
In this section, we analyze the core when searchers find bundles with some probability and success is independent across searchers and opportunities.
The stochastic model assumes binary-valued opportunities. Note the limitation of the setup, as it overlooks the complexity of MEV dynamics, where opportunities are often correlated, and valuations can vary widely.
Let $v_{ij}$ be a binary random variable, which is $1$ with probability $p$ and $0$ with probability $1-p$.\footnote{Generalizations of the subsequent results to non unit value that can be different for different opportunities are straightforward. The only assumption needed is that searchers conditional on finding the same opportunity generate the same value from it. We could also accommodate heterogeneous value from the same opportunity as long as the noise is sufficiently bounded.} Thus, $p$ measures how easy it is to find an arbitrage (bundle). 

\begin{proposition}\label{upperbound}
Let  $n:=|\mathcal{S}|$.
If $p> \frac{2\log n}{n}$ and $m:=|A|<n$, then the validator receives the entire block value with high probability in any core allocation.
\end{proposition}  

\begin{proof}
    We show that there are at least $2$ searchers who have positive value each arbitrage. This can be done using direct computation of probabilities and an application of the union bound inequality. 
    Consider the following sum $Y_i:=\sum_{j\in\mathcal{S}}v_{ij}$. Thus, $Y_i$ is the number of searchers that find opportunity $i$. 
    Note that $Y_i$ is a Binomial random variable with parameters $n$ and $p$, that is $Y_i\sim Bin(n,p)$.
    %Then, by linearity of expectation $\mathbb{E}[Y_i]=np$. By Hoeffding's inequality~\cite{hoeffding_inequality}, we obtain: 

    %\begin{equation}\label{hoeffding_inequality}
     %   P[|Y_i-np| \geq t] \leq 2\exp(\frac{-2t^2}{n}),
    %\end{equation}
    %for any $t > 0$. By plugging in $t=\sqrt{n}\log n $ and $p=\frac{\log n}{\sqrt{n}}$ in~\eqref{hoeffding_inequality}, we obtain

    \begin{eqnarray}\label{bounding_yi}
        P[Y_i<2]&=& P[Y_i=0]+P[Y_i=1] = (1-p)^n+{n\choose 1}(1-p)^{n-1}p \nonumber\\ 
        &\leq& e^{-pn}+npe^{-p(n-1)}\leq \frac{1}{n^2}+2\log n \frac{1}{n^2}\leq \frac{2\log n}{n^2}, 
    \end{eqnarray}
where the first inequality is obtained from the well known inequality: $1-x\leq e^{-x}$ for any $x>0$. By the union bound, we have: $$P[\text{at least one } Y_i < 2]\leq mP[Y_1 < 2] \leq n \cdot \frac{\log n}{n^2} = \frac{\log n}{n}.$$ The last inequality is by~\eqref{bounding_yi}.
    For any $\varepsilon>0$ there is a $n$ large enough so that $\frac{\log n}{n} < \varepsilon$ 
    and therefore $$P[Y_i\geq 2\text{ for any } i]\geq 1-\varepsilon.$$ Applying Proposition~\ref{thm:main} shows the claim of the proposition.
\end{proof} 
On the other hand, if the probability of discovering opportunities shrinks sufficiently fast in the number of searchers, then searchers can capture value with positive non-vanishing probability:
\begin{proposition}\label{1/n}
If $p=\Theta(\frac{1}{n})$ then with positive probability that is constant in $n$ the validator receives $0$ in the searcher-optimal core allocation.  
\end{proposition} 

\begin{proof}
    For each opportunity $i\in \mathcal{A}$, with constant probability, there is exactly one searcher that has a positive value, i.e., $Y_i=1$. Namely, $$P[Y_i=1]={n \choose 1}(1-p)^{n-1}p\rightarrow \frac{1}{e}\Theta(1),$$ as $n\rightarrow \infty$. Then, $P[Y_i=1\text{ for any } i]\rightarrow\Theta\left(\frac{1}{e^m}\right)$ as $n\to\infty$. That is, searchers complement each other in finding different arbitrages. From Proposition~\ref{thm:main}, with probability $\Theta\left(\frac{1}{e^m}\right)$, we have $M=0$.
\end{proof} 
\begin{remark}
The previous results discuss the value distribution for scenarios where the block value is expected to be positive. If probability shrinks faster than $1/n$ as $n$ grows, then with high probability the produced block has value $0$. However, conditional on the block having positive value, all value is captured by searchers in the searcher-optimal core allocation.
\end{remark}
As remarked in footnote~\ref{footnote}, block builders observe around $\sim 1/3$ of submitted bundles clashing. We can use this number to get approximate values of the parameters in our model: we use as $n=125$ which is the number of addresses that have placed at least $2$ bundles in Ethereum blocks within the last $30$ days prior to writing this paper (excluding addresses that have landed only one bundle gives us a crude way to identify addresses that have a high chance of being the main address used by a searcher)~\footnote{According to the website \url{https://libmev.com/}.}. 
Then
$$2/3\approx P[Y_i<2]=(1-p)^{125}+125(1-p)^{124}p\Rightarrow p\approx 1\%.$$
\subsubsection*{Capacity Constraints}
The previous model can be easily adapted to the case of capacity constraints on blocks where a block can only contain up to a fixed number of transactions, denoted by $K$. Note that submodularity is maintained when adding a capacity constraint. Thus, Proposition~\ref{thm:main} naturally extends to the case of capacity constraints. For the model with independent bundles, we now can consider the value function
$$\bar{v}^K(S)=\max_{B\in\mathcal{B}(S)}\sum_{i\in\mathcal{A}}\sum_{j\in\mathcal{S}}v_{ij}B_{ij}=\max_{A'\subseteq \mathcal{A},|A'|\leq K}\sum_{i\in{A}'}\max_{j\in\mathcal{S}}v_{ij}.$$
It follows immediately that Corollary~\ref{value} holds with  $$M^K:=\max_{A'\subseteq\mathcal{A},|{A}'|\leq K}\sum_{i\in{A}'}M_i$$
and $x_j^K:=\sum_{i:j\in\mathcal{A}'\cap\text{argmax}_jv_{ij}}(v_{ij}-M_i).$ 

Next, we consider how the bounds on searcher and validator value capture for the stochastic model change with capacity constraints on the block. We assume that only a constant fraction of possible opportunities can be accommodated. With a bound on the block size, the probability threshold above which the validator captures all value in all core allocation becomes lower, and now matches the corresponding threshold from Proposition~\ref{1/n} where the probability is positive (Proposition~\ref{1/n} still holds with a bound on the block size).

\begin{proposition}
Let  $n:=|\mathcal{S}|$. assume the block has capacity to include $(1-\alpha)m$ transactions where $m:=|\mathcal{A}|$ and $1/m<\alpha<1$ is a constant. Then, there is a decreasing function $\phi$ such that if $p=\omega (\frac{\phi(\alpha)}{n})$, the validator gets the entire block value with high probability in any core allocation.
\end{proposition}  

\begin{proof}
    We show that there are at least $2$ searchers who have positive value each arbitrage. 
    As in the proof of Proposition~\ref{upperbound} defining $Y_i:=\sum_{j\in\mathcal{S}}v_{ij}$, we obtain 
    \begin{eqnarray}\label{bounding_yi2}
        P[Y_i<2]\leq e^{-pn}+npe^{-p(n-1)}. 
    \end{eqnarray}
    Now consider the probability that for more than $\alpha m$ indices we have $Y_i<2$. This is bounded by
\begin{eqnarray*}
    P[Y_i < 2 \text{ for at least }\alpha m \text{ indices} i]&\leq& {m \choose \alpha m}P[Y_1 < 2]^{\alpha m} \\
    &\leq& \left(\frac{e}{\alpha}\left(e^{-pn}+pne^{-p(n-1)}\right)\right)^{\alpha m},
\end{eqnarray*}
where the last inequality uses the well-known inequality ${a \choose b}\leq\left(\frac{ae}{b}\right)^b$ and the previously obtained inequality~\eqref{bounding_yi2}.
     Choosing $\phi(\alpha)$ to satisfy
     $$(1+\phi(\alpha))e^{-\phi(\alpha)}=\frac{\alpha}{e},$$
     we have for $p=\omega(\phi(\alpha)/n)$ that for any $\epsilon>0$ there is a $n$ such that
    $$\frac{e}{\alpha}\big(e^{-pn}+pne^{-p(n-1)}\big)<\epsilon$$
    and therefore $$P[Y_i\geq 2\text{ for at least }(1-\alpha)m\text{ indices } i]\geq 1-\varepsilon^{\alpha m}>1-\varepsilon.$$ Applying Proposition~\ref{thm:main} shows the claim of the proposition.
\end{proof} 
\begin{remark}
The previous result is tight in the sense that we have constant positive value capture by searcher for $p=\Theta(1/n)$, Proposition~\ref{1/n}, still holds for the case of capacity constraints on the block.
\end{remark}
\section{Empirical validation}
Our model predicts a positive relation between the amount extracted by the proposer and the competition between searchers. To verify this theory, we can look at aggregate data on submitted bundles from the MEV-Share OFA (Order Flow Auction)\footnote{\url{https://docs.flashbots.net/flashbots-protect/mev-share}}. 
\begin{figure}
\includegraphics[scale=0.22]{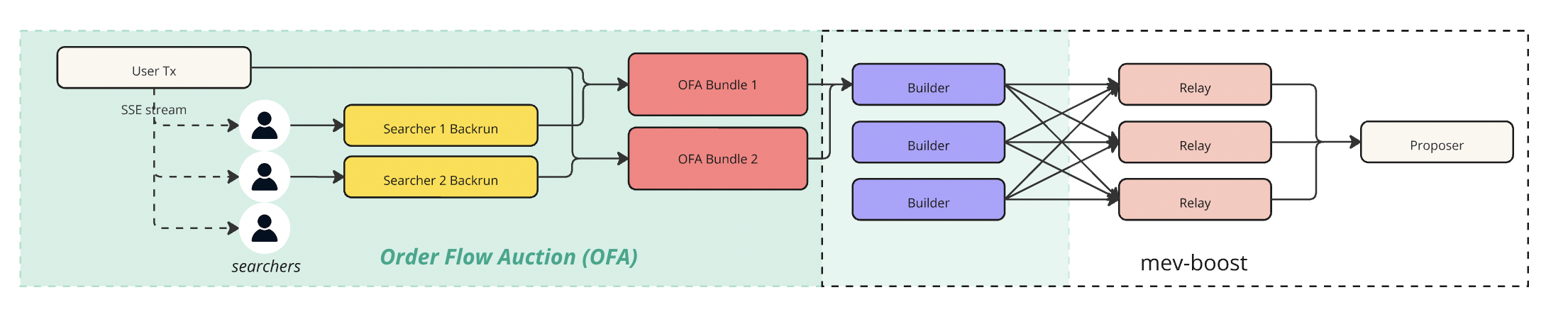}

\caption{A schematic representation of an order flow auction.}\label{OFA}
\end{figure}
As illustrated in Figure~\ref{OFA}, a general OFA exposes user transactions to a set of searchers for potential backrun opportunities, and forward the valid bundles to builders to settle, with profits shared from searchers as refunds to users and profits to the proposer. MEV-Share allows users to submit transactions to a private mempool, hiding (part of) the transaction information from searchers. Searchers can backrun user transactions submitted to the OFA based on the information given by the user, and the OFA simulates the user transaction together with the searcher backrun as a bundle, to check for feasibility. When submitting its transaction to OFA, the user can decide how much profit shared from the searcher will be refunded to itself vs paid to the proposer, with the trade-off that the higher the share it wants to keeps for itself, the less value it can provide to a builder to be considered in the next block. Based on this profit-sharing functionality, bundles created on MEV-Share are different from other bundles, in the sense that the value that is not captured by the searcher goes in different shares to the proposer \emph{and} the user who sent the target transaction rather than to the proposer alone. However, for the purpose of validating our theory, we can look at bundle profits jointly captured by the proposer and the original transaction senders and verify whether it is increasing in the number of bundles submitted by searchers with the same target user transaction.

We consider bundle data for the time period between September 2023 and February 2024. Bundles generated through MEV-Share make up a significant fraction of the bundles submitted to the Flashbots block builder. We only considered feasible submitted backrun bundles. Figure~\ref{fig:enter-label} displays the frequency of different numbers of backrun bundles submitted for the same target transactions. 
We observe that the distribution is skewed towards $0$. This is in line with our model if we choose parameters $p$ and $n$ such that the expectation of finding a successful backrun is low, $pn<1$.

\begin{figure}[t]
    \centering
    \includegraphics[scale=0.2]{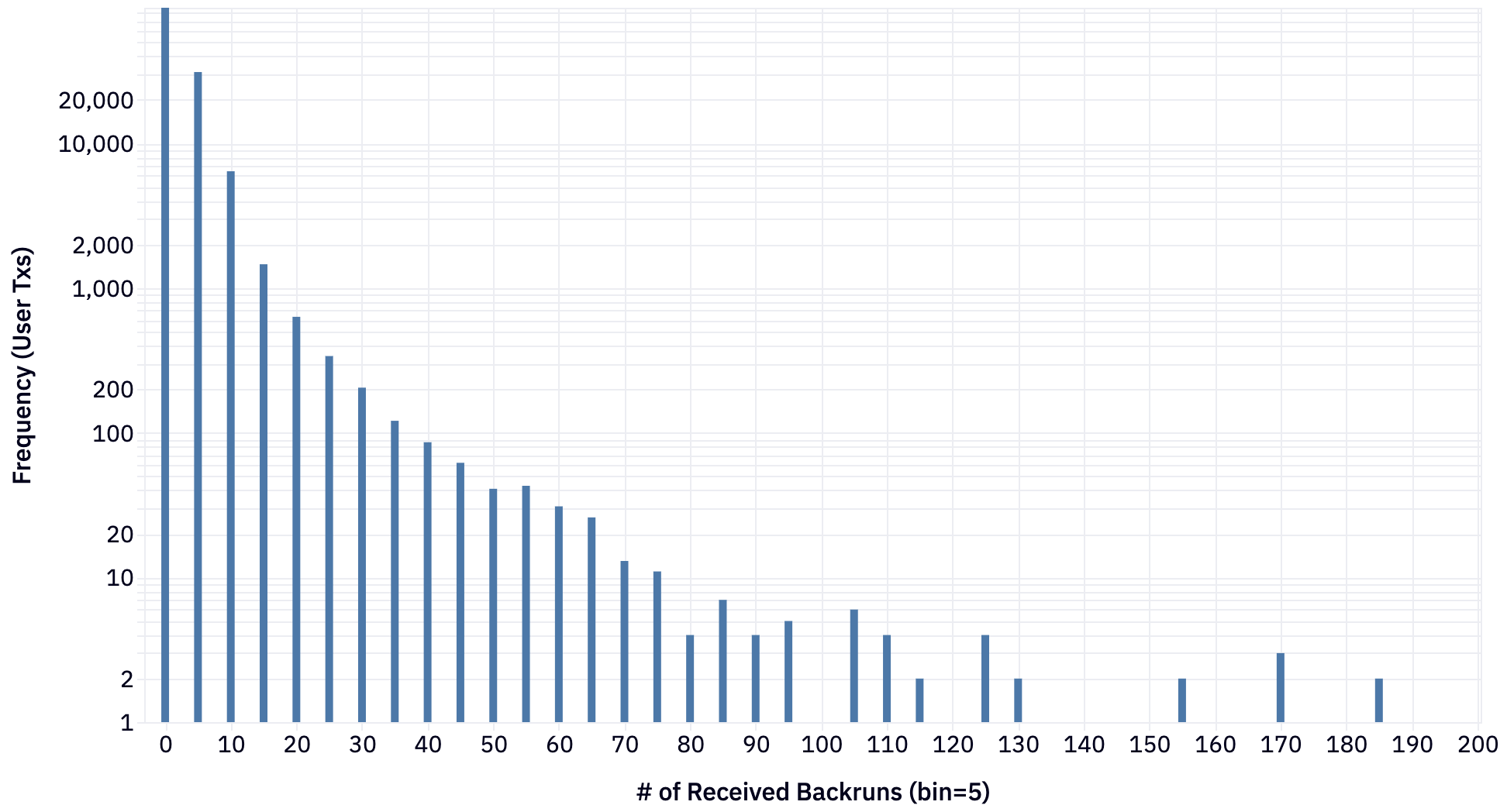}
    \caption{Number of submitted bundles for the same target transaction. Data is grouped into bins of 5, so that the first bar corresponds to the number of transactions receiving 0-4 backruns, the second to the number of transactions receiving 5-9 backruns etc.}
    \label{fig:enter-label}
\end{figure}
In Figure~\ref{fig:enter-label2}, we display the median profit for validators/users as a function of the number of submitted bundles for the same target transaction. We observe a positive relation between the number of searchers targeting the same transaction and the proposer/user profit for that transaction. This is generally in line with our theoretical predictions. We have to make one important caveat though: the value generated by the searcher is not always directly observable in our data set. This, is because the mechanism used by MEV-share is not eliciting the true value of the searchers (in contrast to the VCG mechanisms that we have analyzed above). Instead, the mechanism can be interpreted as a first-price sealed-bid auction for inclusion of the bundle. For back-runs sometimes the value accrued of the searcher can be inferred by his submitted other transactions, but not always and not perfectly. As the auction is first-price, bidding the full true value is not optimal for the searcher, and we are in a situation of partially hidden information.
There is also correlation between the value of the opportunity and the number of searchers competing for it, so that higher bids can be partially explained by higher value. Second highest values among submitted bundles is, however, close to the highest value.  Thus, the difference in profit going from 1 to 2 bidders can be explained by competition, as most value is extracted by validator/user if there are two or more bundles competing.
\begin{figure}[t]
    \centering
    \includegraphics[scale=0.2]{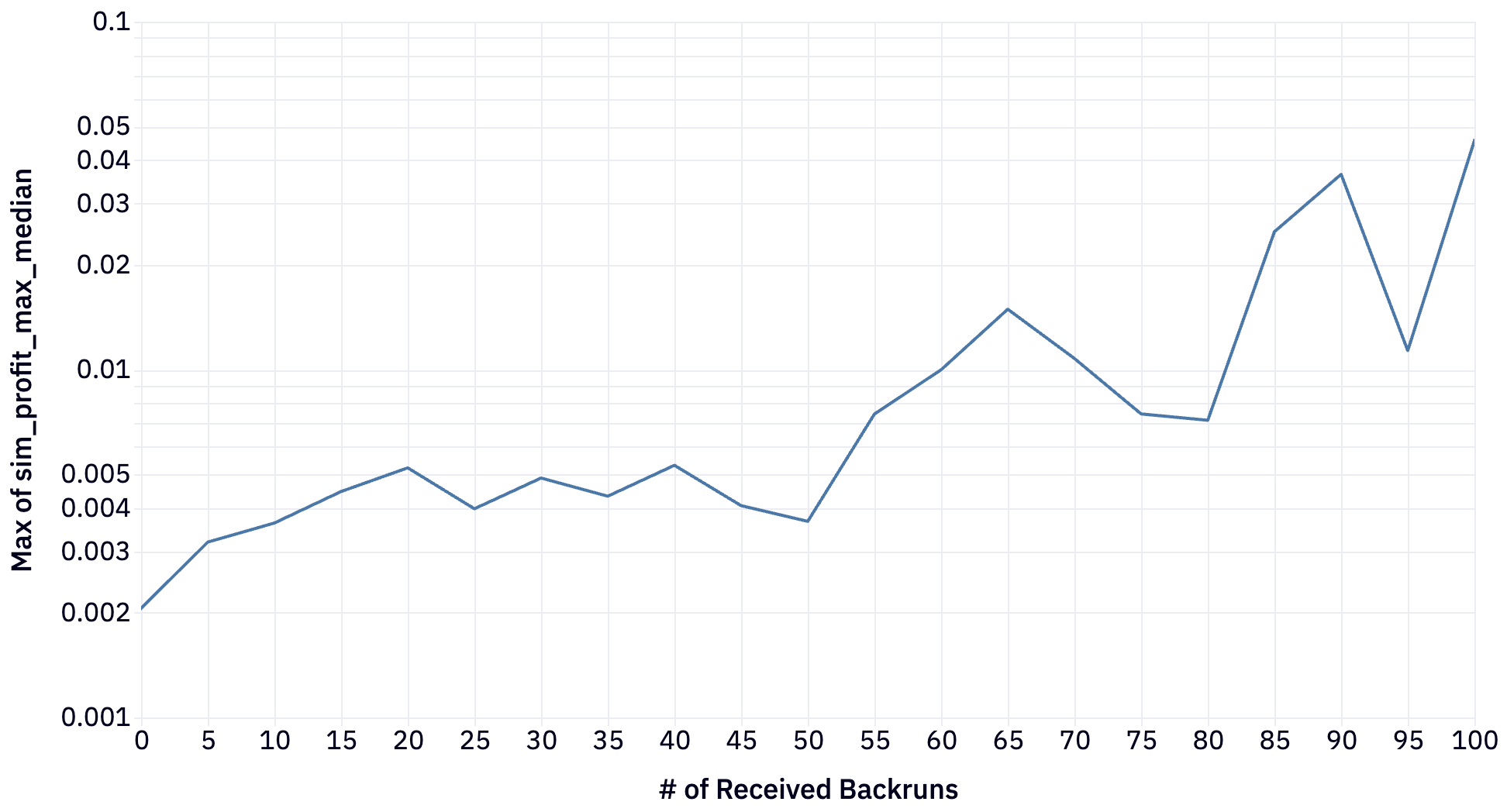}
    \caption{Median profit (log-scale) in ETH as a function of number of submitted bundles for the same target transaction.}
    \label{fig:enter-label2}
\end{figure}
Our model predicts very sharply that if values can be elicited perfectly, going from one to two competing searchers for the same target transaction, all value from that opportunity would go to the builder/searcher. As there is potentially hidden information, this sharp prediction, does not need to hold in our data however. But we would still expect the profit for the validator/users to to increase sharply going from 1 to 2 bidders, and there to be a diminishing effect of adding additional bidders. Zooming in, on the data, we indeed observe a jump at $1$, see Figure~\ref{fig:enter-label3}. Going from 1 to 2 searchers competing for the same opportunity raises the profit more than any additional searcher, in line with the theory and the previous considerations.
\begin{figure}[t]
    \centering
    \includegraphics[scale=0.18]{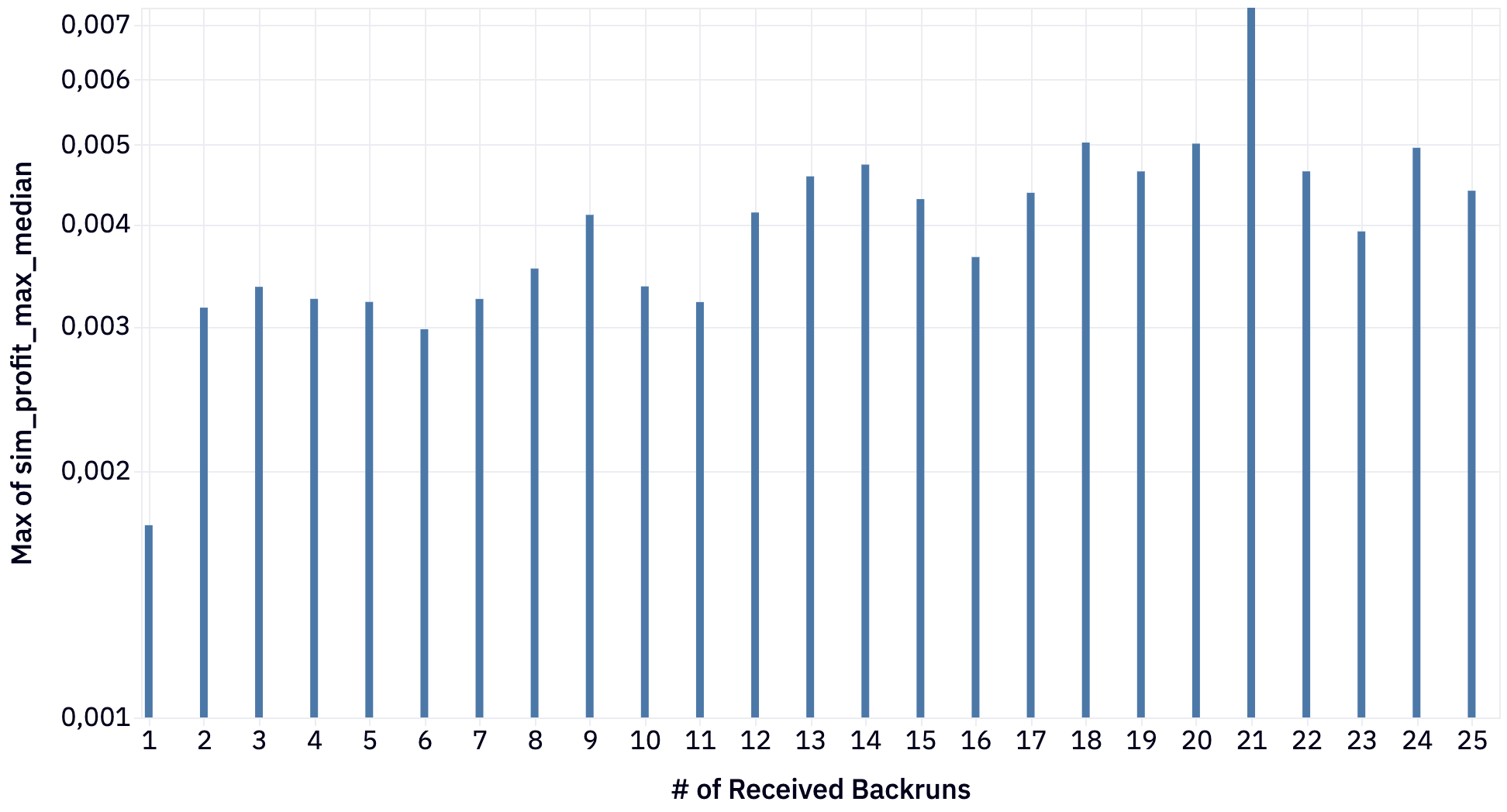}
    \caption{Median profit (log-scale) in ETH as a function of number of submitted bundles for the same target transaction .}
    \label{fig:enter-label3}
\end{figure}

Finally, we run a linear regression, see Table~\ref{regression}, of the logarithmic median profits on the number of bundles submitted for the same target transaction. Although the number of observations is not very large, the regression gives a statistically significant positive relation (triple stars mean that the probability hypothesis is wrong is less than 0.001) between profits and the number of submitted bundles. An additional searcher competing for an opportunity raises profits of validator/users by $1.1\%$ on average.
\begin{table}[h]
\caption{}
\label{regression}
\begin{center}
\begin{tabular}{ll}
\hline
                     & OLS        \\
\hline
const               & -6.087***  \\
                    & (0.107)    \\
received\_backrun   & 0.011***   \\
                    & (0.002)    \\
R-squared           & 0.340      \\
No. of observations & 110        \\
%F-statistic         & 38.56      \\
\hline
\end{tabular}
\end{center}
\end{table}

\section{Multiple Concurrent Proposers}
In this section, we will look at a setting with multiple concurrent proposers in which every next block is built from the input of designated set of proposers. Formally, similar to our approach so far, we can model it with finite set $\mathcal{V}$ of validators, instead of one validator. Let $\mathcal{V}_{-i}$ denote the set of all validators without the validator indexed $i$. We make the mild assumption that the set $\mathcal{V}_{-i}$ can build a block that includes all potential opportunities (bundles), for any $i\in \mathcal{V}$. This assumption can be easily justified if the block is large enough, compared to how much gas opportunities in that block time need, and the total number of validators, $|\mathcal{V}|$, is large enough.  Formally:

\begin{equation}
    v(S\cup \mathcal{V}_{-j}) = \max_{B \in \mathcal{B}}(\sum_{i\in S}v_i(B)+\sum_{p\in \mathcal{V}_{-j}}v_p(B)). 
\end{equation}

Moreover, we assume that proposers themselves do not derive any value from building any block, i.e., $v_p(B)=0$ for any $p\in \mathcal{V}$ and $B\in \mathcal{B}$. This is a reasonable assumption when proposers only derive value transferred to them from the searchers. 

With only this assumption, we can obtain the following negative result:

\begin{proposition}
    The core of the game is empty when $|\mathcal{S}|>1$.
\end{proposition}

\begin{proof}
    First, note that no validator can get a positive value in the core allocation. The proof is by contradiction. Suppose that the validator indexed $j$ gets a positive allocation, that is, $x_j>0$. Then, there is a deviating subset of all other validators $\mathcal{V}_{-j}$ and all searchers $\mathcal{S}$, in which they get allocated all the value, higher than the original allocation, a contradiction to the core condition that there is no deviation that has a higher total value compared to the total allocation.
\end{proof}

However, if there is only one searcher, $|\mathcal{S}|=1$, the core is non-empty.

\begin{proposition}
If $|\mathcal{S}|=1$, then the core consists of the allocation in which the unique searcher gets all the value.
\end{proposition}

\begin{proof}
    First, we show that the allocation from the proposition claim is in the core. There is no deviating subset of proposers, since they cannot create any value themselves.
    Next, note that if any proposer was getting a positive allocation, then the subset consisting of the only searcher and all the other proposers would be deviating and getting more total allocation.
\end{proof}
The emptiness of the core for the multiple concurrent builder setting points towards an inherent instability of such a setting. In the presence of multiple proposers we would expect frequent updates and cancellations of searchers' bids until the end of the block building window, as the market is out-of-equilibrium and re-contracting is always profitable for some players.

\section{Conclusion}
We have studied MEV extraction in block building as a function of searcher competition and have argued that the core, and in particular the rule that selects the searcher-optimal point within the core are suitable solution concepts that allow us to make robust prediction about value distribution in MEV extraction. We have further identified a dominant-strategy incentive compatible mechanism, giving searchers their marginal value contribution to the realized block, which would be a theoretically appealing payment mechanism for searchers in the case of passive proposers/builders. The model has been extended to multiple concurrent proposer setting, in which we have obtained that the core is empty.

The model and solution concept allowed us to make sense of stylized facts about the Ethereum block-building market: validators capture most of the value most of the time, and searchers with unique edge that are less exposed to competition are able to capture significant value. A natural extension of our model for further research would add correlation between different MEV opportunities to our stochastic model. Such an enhanced model would be particularly suitable to study the competition between different block builders and would possibly make theoretical predictions about concentration and value capture in the builder market.

\bibliographystyle{plain}
\bibliography{references}
\end{document}